\documentclass[12pt, onecolumn, draftcls]{IEEEtran}

\usepackage{epsfig, psfrag,latexsym,amsmath,epsf,amssymb,bm}
\usepackage{cite}

\newtheorem{theorem}{Theorem}

\begin{document}
\IEEEoverridecommandlockouts

\title{The Capacity Region of the Cognitive Z-interference Channel with One Noiseless Component}
\author{%
  \authorblockN{Nan Liu\authorrefmark{1},
     Ivana Mari\'c\authorrefmark{1},
     Andrea J. Goldsmith\authorrefmark{1},
     Shlomo Shamai (Shitz) \authorrefmark{2}
  }\\
  \authorblockA{%
  \authorrefmark{1}Dept.\ of Electrical Engineering,
                     Stanford University, Stanford, CA 94305
  }
    \authorblockA{%
    \authorrefmark{2}Department of Electrical Engineering, Technion, Technion City,
Haifa 32000, Israel\\
  }
  Email:  \texttt{\ {\small \{nanliu@stanford.edu, ivanam@wsl.stanford.edu, andrea@wsl.stanford.edu,  sshlomo@ee.technion.ac.il\} } }
\thanks{The work of N. Liu, I. Mari\'c and A. J. Goldsmith was supported in part from
the DARPA ITMANET program under grant 1105741-1-TFIND, Stanford's
Clean Slate Design for the Internet Program and the ARO under MURI
award W911NF-05-1-0246. The Work of S. Shamai was supported by the
ISRC Consortium and by  the European Commission in the framework of
the FP7 Network of Excellence in Wireless COMmunications NEWCOM++.}
}
\date{}
\maketitle

\begin{abstract}
We study the discrete memoryless Z-interference channel (ZIC) where
the transmitter of the pair that suffers from interference is
cognitive.  We first provide upper and lower bounds on the capacity
of this channel. We then show that,  when the channel of the
transmitter-receiver pair that does not face interference is
noiseless, the two bounds coincide and therefore yield the capacity
region. The obtained results imply that, unlike in the Gaussian
cognitive ZIC, in the considered channel  superposition encoding at
the non-cognitive transmitter as well as Gel'fand-Pinsker encoding
at the cognitive transmitter are needed in order to minimize the
impact of interference. As a byproduct of the obtained capacity
region, we obtain the capacity result for a generalized
Gel'fand-Pinsker problem.
\end{abstract}

\vspace*{0.4cm} \noindent {\em Index terms:} Cognitive interference
channel, capacity region, Z-interference channel, one-sided
interference channel

\newpage

\section{Introduction}
The interference channel (IC) \cite{Shannon:1961} is a simple
network consisting of two transmitter-receiver pairs. Each pair
wishes to reliably communicate at a certain rate, however, the two
communications interfere with each other. A key issue in such
scenarios then,  is  how to handle the interference introduced by
the simultaneous transmissions.  This issue is not yet fully
understood,  and the problem of finding the capacity region of the
IC  remains open, except in special
cases\cite{Ahlswede:1971,Sato:1977,Carleial:1978,Benzel:1979,ElGamal:1982,Carleial:1975,
Sato: 1978MayIC, Han: 1981, Sato:1981,Costa:1987,
Liu_Ulukus:2006Allerton}. For a tutorial on the capacity results of
the IC, see \cite{Kramer:2006}. The Z-interference channel (ZIC) is
an IC where one transmitter-receiver pair is interference-free.
Although this is a simpler channel model than the IC, capacity
results are still known only in special cases \cite[Section
IV]{ElGamal:1982}, \cite{Sason:2004,Ahlswede:2006,
Liu_Goldsmith:2008ISIT}.

In certain communication scenarios, such as cognitive radio
networks, some transmitters are cognitive, i.e., are able to sense
the environment and thus obtain side information about transmissions
in their vicinity. Perhaps due to the exciting promise of the
cognitive radio technology to improve the bandwidth utilization and
thus allow for new wireless services and a higher quality of
service, the IC with one cognitive transmitter has been studied
extensively \cite{Devroye:2006, MYK2007, WuVishwanath2006,
JovicicViswanath2006, MGKSETT2007, JiangXin2007, Cao:2008,
SridharanVishwanath2007}. Related  channel models were also analyzed
in \cite{LiangBaruchPoorShamaiVerdu2007, Cao:2007}. In the model
considered in \cite{Devroye:2006, MYK2007, WuVishwanath2006,
JovicicViswanath2006, MGKSETT2007, JiangXin2007, Cao:2008,
SridharanVishwanath2007, LiangBaruchPoorShamaiVerdu2007}, it is
assumed that due to the cognitive capabilities,  the cognitive
encoder noncausally obtains the full message of the non-cognitive
transmitter. While this is a somewhat idealistic view of cognition
in a wireless network, this model
 applies  for example, to scenarios where the cognitive transmitter is a base station. Then, it can obtain side information via backhaul (high-capacity link such as an optical cable). This side information then  enables interference reduction \cite{Gelfand:1980} by precoding at the cognitive encoder. Furthermore, it enables cooperation with the non-cognitive pair. In fact, one of the main difficulties in finding the capacity region of the traditional IC comes from distributed encoding. IC with one cognitive transmitter enables one-sided transmitter cooperation, and thus allows centralized encoding to some degree. This may be the reason why determining the capacity region of the cognitive IC is somewhat easier than the traditional IC. In particular, while the capacity region of the Gaussian IC in weak interference is not known (the sum capacity in certain weak interference regimes has recently been found in \cite{Shang:2008, Anna:2008, Motahari:2008}), the capacity region of the \emph{cognitive} Gaussian IC in weak interference has been determined \cite{WuVishwanath2006, JovicicViswanath2006}.

In this paper, we study a ZIC where the transmitter of the pair that
suffers from interference is cognitive (see Fig.~\ref{Fig1}). The
capacity region of such a cognitive ZIC in the Gaussian case is
straightforward to obtain, since by using dirty-paper coding
\cite{Costa:1983} at the cognitive encoder, both communicating pairs
can  achieve the interference-free, single-user rates. However,
limiting the study of the cognitive ZIC to the Gaussian case leaves
some unsatisfaction to the understanding of the problem. Firstly, it
does not provide intuition as to how the interferer's rate affects
the rate of the cognitive transmitter-receiver pair in a general
channel. Secondly, it does not provide the insight into the optimal
codebook structure for the non-cognitive encoder, so that it
minimizes interference caused for the cognitive pair.

Hence, in this paper, we study a \emph{discrete} memoryless
cognitive ZIC. We first derive an upper bound on the capacity
region. The technique that we use in obtaining the converse was
introduced by Korner and Marton in \cite{Korner:1977Image}, and was
proven to be useful in the solution of several problems in
multi-user information theory \cite{Korner:1977,Korner:1977Image,
Ahlswede:2006,Liu_Goldsmith:2008ISIT}, including  the
Gel'fand-Pinsker problem \cite{Gelfand:1980}. We apply this
technique \emph{twice} to obtain the upper bound on the capacity
region. Next, we derive a lower bound on the capacity region where
the non-cognitive pair uses superposition encoding to control the
amount of interference it causes for the cognitive pair. Unlike in
the IC, this encoding approach has not been applied in the cognitive
IC literature, with the exception of concurrent and independent work
\cite{Cao:2008}. Finally, we show that the lower and upper bounds
meet when the channel between the non-cognitive pair is noiseless.
We  denote this channel model as the cognitive ZIC with one
noiseless component. From the capacity results, we conclude that it
is optimal for the interference-causing (non-cognitive) pair to use
superposition encoding; the inner codeword is decoded by the
receiver of the cognitive pair while Gel'fand-Pinsker coding is
performed against the outer codeword at the cognitive transmitter.
\begin{figure}[t]
\center\epsfig{figure=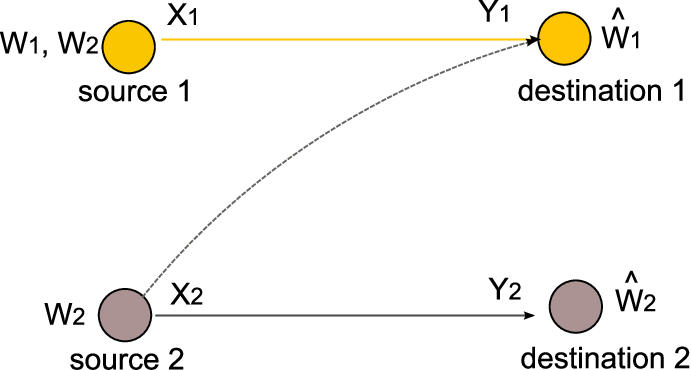,width=7cm,
height=3cm} \caption[]{Cognitive Z-interference channel.}
      \label{Fig1}
\end{figure}

The capacity region of the discrete memoryless cognitive IC is known
in some special cases \cite{WuVishwanath2006, MYK2007, Cao:2008}.
The tight result we derive in this paper does not fall into these
special cases, as explained in more details in Section
\ref{capacity}. Furthermore, the cognitive ZIC with one noiseless
component is the first channel model for which superposition
encoding at the non-cognitive transmitter is not only required but
also optimal.

Note that, in general, the capacity region of the traditional ZIC in
which the interference-free transmitter-receiver pair is noiseless,
is unknown. The most we know about this scenario is the sum capacity
\cite{Ahlswede:2006}. Thus, the results in this paper provide yet
another example where finding the capacity region of the cognitive
IC is easier than that of the traditional IC due to the possibility
of centralized (joint) encoding by  the cognitive transmitter.

The considered problem is also intimately related to the
Gel'fand-Pinsker (GP) problem \cite{Gelfand:1980} where a
transmitter-receiver pair communicates in the presence of
interference noncausally known at the encoder (see
Fig.~\ref{GPproblem}). By viewing the non-cognitive encoder in the
cognitive ZIC as a source of this interference, we arrive to a
generalized GP  problem. Instead of the state being i.i.d. as in the
GP problem, in the generalized GP model considered in this paper,
the state is uniformly distributed on a set of size $2^{nR_2}$,
where $R_2$ is a number between $0$ and the logarithm of the
cardinality of the state space. The further generalization is  that,
unlike in \cite{Gelfand:1980}, in our model one can optimize the
set, i.e., the structure of the interference.  The solution of this
paper shows that the optimal interference has a superposition
structure.

\section{System Model}
Consider a ZIC with two transition probabilities $p(y_1|x_1,x_2)$
and $p(y_2|x_2)$. The input and output alphabets are
$\mathcal{X}_1$, $\mathcal{X}_2$, $\mathcal{Y}_1$ and
$\mathcal{Y}_2$.

Let $W_1$ and $W_2$ be two independent messages uniformly
distributed on $\{1,2,\cdots,M_1\}$ and $\{1,2,\cdots, \break
M_2\}$, respectively. Transmitter $i$ wishes to send message $W_i$
to Receiver $i$, $i=1,2$. Transmitter 1 is cognitive in the sense
that, in addition to knowing $W_1$, it knows the message $W_2$. An
$(M_1,M_2,n,\epsilon_n)$ code for this channel consists of a
sequence of two encoding functions
\begin{align}
f_1^n:& \{1,2,\cdots,M_1\} \times \{1,2,\cdots, M_2\} \rightarrow \mathcal{X}_1^n,\\
f_2^n:& \{1,2,\cdots,M_2\} \rightarrow \mathcal{X}_2^n,
\end{align}
and two decoding functions
\begin{align}
g_i^n: \mathcal{Y}_i^n \rightarrow \{1,2,\cdots,M_i\}, \qquad i=1,2
\end{align}
with probability of error
\begin{align}
\epsilon_n=\max_{i=1,2} \quad \frac{1}{M_1 M_2} \sum_{w_1,w_2}
\text{Pr} \left[g_i^n(Y_i^n) \neq w_i|W_1=w_1,W_2=w_2 \right].
\end{align}
A rate pair $(R_1,R_2)$ is said to be achievable if there exists a
sequence of $\left(2^{nR_1}, 2^{nR_2}, n, \epsilon_n \right)$ codes
such that $\epsilon_n \rightarrow 0$ as $n \rightarrow \infty$. The
capacity region of the cognitive ZIC is the closure of the set of
all achievable rate pairs.
\begin{figure}[t]
\center\epsfig{figure=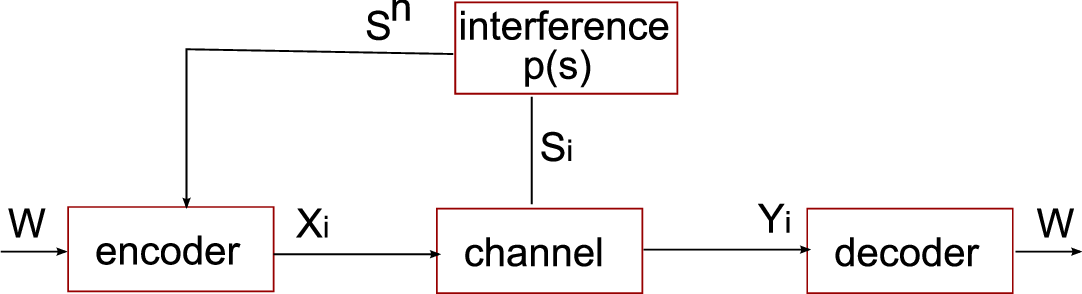,width=8cm, height=2.5cm}
\caption[]{Gel'fand-Pinsker problem.}
      \label{GPproblem}
\end{figure}

A cognitive ZIC with one noiseless component is a cognitive ZIC
where the channel between $X_2$ to $Y_2$ is noiseless, i.e.,
$p(y_2|x_2)$ is a deterministic one-to-one function.

Throughout the paper, we use the following shorthand for random
vectors: $K^i\overset{\triangle}{=}K_1, \break K_2, \cdots, K_i$ and
$K_{i+1}^n \overset{\triangle}{=}K_{i+1}, K_{i+2}, \cdots, K_n$.

\section{Converse}
In this section, we provide an upper bound on the capacity region of
the cognitive ZIC.

\begin{theorem} \label{thmconverse}
Achievable rate pairs $(R_1, R_2)$  belong to a union of rate
regions given by
\begin{align}
R_1& \leq I(U;Y_1|V)-I(U;Y_2|V) \label{conresult1}\\
R_2 & \leq I(X_2;Y_2|V)+ \min \left\{ I(V;Y_1), I(V;Y_2) \right\}
\label{conresult2}
\end{align}
where the union is over all probability distributions
$p(v,u,x_2)p(x_1|u,x_2)$ and the mutual informations are calculated
according to distribution
\begin{align}
p(v,u,x_1,x_2,y_1,y_2)=p(v,u,
x_2)p(x_1|u,x_2)p(y_1|x_1,x_2)p(y_2|x_2). \label{dis}
\end{align}
\end{theorem}
\begin{proof}
The proof is provided in Section \ref{proof_thmconverse}.
\end{proof}
This converse result is obtained by using the converse technique of
Korner/Marton \cite[page 314]{Csiszar:book} two times, resulting in
two auxiliary random variables.

\section{Achievability}
The achievability scheme uses a combination of superposition
encoding at the non-cognitive encoder and GP encoding of the outer
codeword of interference at the cognitive encoder. The performance
is given in the following theorem.

\begin{theorem} \label{thmach}
The union of rate regions given by
\begin{align}
R_1& \leq I(U;Y_1|V)-I(U;X_2|V) \label{achresult1}\\
R_2 & \leq I(X_2;Y_2|V)+ \min \left\{ I(V;Y_1), I(V;Y_2) \right\}
\label{achresult2}
\end{align}
is achievable, where the union is over all probability distributions
$p(v,u,x_2)p(x_1|u,x_2)$ and the mutual informations are calculated
according to the distribution in  (\ref{dis}).
\end{theorem}
\begin{proof}
The proof is provided in Section \ref{proof_thmach}.
\end{proof}
{\bf Remark:} The proposed achievability scheme is a special case of
the independent and concurrent work  \cite[Theorem 2]{Cao:2008} by
setting $U_{10}=V$, $(U_{11}, U_{10})=(V_{11}, U_{10})=X_1$,
$V_{20}=\phi$, $V_{22}=U$, $L_{20}=R_{20}=0$, $L_{11}=R_{11}$ and
then swapping the indices of $1$ and $2$ because in \cite{Cao:2008},
the \emph{second} transmitter-receiver pair is cognitive.

\section{Capacity Region of the Cognitive ZIC with One Noiseless Component} \label{capacity}
In general, the achievability results in
(\ref{achresult1})-(\ref{achresult2}) and the converse results in
(\ref{conresult1})-(\ref{conresult2}) do not meet, due to the fact
that
\begin{align}
I(U;X_2|V) \geq I(U;Y_2|V)
\end{align}
because the random variables satisfy (\ref{dis}) which implies that
Markov chain $U \rightarrow (V,X_2) \rightarrow Y_2$ holds.

However, in the case where the channel output $Y_2=X_2$, the
achievability results and the converse results meet, yielding the
capacity region. More specifically, we have the following capacity
results for the cognitive ZIC.

\begin{theorem} \label{thmcapacity}
For cognitive ZIC with one noiseless component,  i.e., $p(y_2|x_2)$
is a deterministic one-to-one function, the capacity region is given
by the union of rate regions:
\begin{align}
R_1& \leq I(U;Y_1|V)-I(U;X_2|V) \label{capcap1} \\
R_2 & \leq H(X_2|V)+ \min \left\{ I(V;Y_1), I(V;X_2) \right\}
\label{capcap2}
\end{align}
where the union is over all probability distributions
$p(v,u,x_2)p(x_1|u,x_2)$ and the mutual informations are calculated
according to the distribution in (\ref{dis}).
\end{theorem}
{\bf Remark:} Similar to the solution of the GP problem, one may
restrict $p(x_1|u,x_2)$ to be a deterministic function, i.e.,
$p(x_1|u,x_2)$ only takes the values of $0$ and $1$, in the union in
Theorem \ref{thmcapacity}. To see this, observe that for a fixed
$p(v,u,x_2)$, only (\ref{capcap1}) and the term $I(V;Y_1)$ in
(\ref{capcap2}) depend on $p(x_1|u,x_2)$. The right-hand side of
(\ref{capcap1}) can be written as
\begin{align}
I(U;Y_1|V)-I(U;X_2|V)=\sum_{v} p(v) \left( I(U;Y_1|V=v)-I(U;X_2|V=v)
\right)
\end{align}
which is a linear combination of convex functions of $p(x_1|u,x_2)$
\cite[Proposition 1 (ii)]{Gelfand:1980}. Thus, the right-hand side
of (\ref{capcap1}) is a convex function of $p(x_1|u,x_2)$ and the
maximum is achieved by a deterministic function. For the fixed
$p(v)$, $I(V;Y_1)$ is convex in $p(y_1|v)$ \cite[Theorem
2.7.4]{Cover:book}, which is a linear function of $p(x_1|u,x_2)$ for
fixed $p(u,x_2|v)$ and $p(y_1|x_1,x_2)$. Thus, $I(V;Y_1)$ is a
convex function of $p(x_1|u,x_2)$ and the maximum is achieved by a
deterministic function. Hence, both (\ref{capcap1}) and $I(V;Y_1)$
in (\ref{capcap2}) are maximized by a deterministic $p(x_1|u,x_2)$.

We conclude from Theorem \ref{thmcapacity} that, in the special case
of noiseless channel between the interference-free
transmitter-receiver pair, to minimize the effect of interference
caused to the cognitive transmitter-receiver pair, the non-cognitive
pair uses superposition encoding, allowing the cognitive pair to
decode the inner codeword. In contrast to the Han-Kobayashi scheme
\cite{Han:1981} for the traditional ZIC, where the outer codeword of
the interferer is treated as noise, here, due to the cognitive
capability of the transmitter that faces interference, GP encoding
is performed on the outer codeword to further reduce the effect of
interference.

The capacity region of the discrete memoryless cognitive IC is known
in some special cases \cite{WuVishwanath2006, MYK2007, Cao:2008}.
The cognitive ZIC with one noiseless component is not a special case
of \cite[Theorem 3]{MYK2007} as it does not satisfy either of the
two conditions of strong interference. It satisfies Assumption 3.1
but not Assumption 3.2 in \cite{WuVishwanath2006}, and therefore its
capacity region is not characterized by \cite[Theorem
3.4]{WuVishwanath2006}. The capacity results in Theorem
\ref{thmcapacity} is not a special case of \cite[Theorem
5]{Cao:2008} as the received signal of the cognitive pair is not a
deterministic function of the two channel inputs. Rather, in the
cognitive ZIC with one noiseless component, the received signal of
the non-cognitive pair is a deterministic function. Furthermore, it
does not satisfy the mutual information inequality required in
\cite[Theorem 5]{Cao:2008}.


\section{Discussion}
In the case where $Y_2=X_2$, the cognitive ZIC problem can be seen
as a form of generalized GP problem, where $X_2$ is the channel
state that affects the communication between transmitter-receiver
pair 1. This formulation generalizes the GP problem in the sense
that, instead of the state (random parameters of the channel) being
i.i.d., the state is uniformly distributed on a set of size
$2^{nR_2}$. Furthermore, we are allowed the freedom, not only to
design the codebook of the cognitive transmitter, but also the
structure of the set where the states lie, in order to maximize the
number of bits transmitted between the cognitive pair. We are then
interested in the capacity of the cognitive transmitter-receiver
pair, denoted as $C(R_2)$, which is a function of $R_2$.

Using the capacity region for the cognitive ZIC with one noiseless
component in Theorem \ref{thmcapacity},  we see that the capacity of
the cognitive pair when the state uniformly takes a value from a set
of $2^{nR_2}$ sequences is
\begin{align}
C(R_2) = \max_{p(v,u,x_2), p(x_1|u,x_2)} I(U;Y_1|V)-I(U;X_2|V)
\label{generalizedGP}
\end{align}
where the maximum is over all distributions $p(v,u,x_2),
p(x_1|u,x_2)$ that satisfy
\begin{align}
H(X_2|V)+\min \left \{ I(V;Y_1), I(V;X_2) \right \} \geq R_2
\label{c}
\end{align}
Thus, in the generalized GP problem, when given the rate of the
possible channel states $R_2$, the optimal interference has a
superposition structure.

\noindent {\bf Remark:} When $R_2=\log |\mathcal{X}_2|$, $C(R_2)$
reduces to the GP rate where the state is i.i.d. and uniformly
distributed on set $\mathcal{X}_2$. This can be seen as follows:
first, by choosing $V=\phi$ and $p(x_2)$ to be the uniform
distribution on $\mathcal{X}_2$ in the maximization of
(\ref{generalizedGP}), we obtain the GP rate. Hence, we conclude
that $C(\log |\mathcal{X}_2|)$ is no smaller than the GP rate. On
the other hand, when $R_2=\log |\mathcal{X}_2|$, according to
(\ref{c}), the distribution we are allowed to maximize over in
(\ref{generalizedGP}) has to satisfy
\begin{enumerate}
\item $p(x_2)$ is the uniform distribution on $\mathcal{X}_2$
\item $I(V;Y_1) \geq I(V;X_2)$
\end{enumerate}
which means
\begin{align}
C(\log |\mathcal{X}_2|) &= \max_{p(v,u|x_2)p(x_1|u,x_2):\hspace{0.1in} I(V;Y_1) \geq I(V;X_2)} I(U;Y_1|V)-I(U;X_2|V) \label{equal1}\\
& \leq \max_{p(v,u|x_2)p(x_1|u,x_2):\hspace{0.1in} I(V;Y_1) \geq I(V;X_2)} I(V;Y_1)+I(U;Y_1|V)-I(V;X_2)-I(U;X_2|V) \label{equal2}\\
&=\max_{p(v,u|x_2)p(x_1|u,x_2):\hspace{0.1in} I(V;Y_1) \geq I(V;X_2)} I(U,V;Y_1)-I(U,V;X_2) \label{equal3}\\
& \leq \max_{p(v,u|x_2)p(x_1|u,v,x_2)} I(U,V;Y_1)-I(U,V;X_2)
\label{equal4}
\end{align}
where in (\ref{equal1})-(\ref{equal4}), we have implicitly assumed
that $p(x_2)$ is the uniform distribution. By setting
$(U,V)=\bar{U}$ in (\ref{equal4}), we see that (\ref{equal4}) is the
GP rate, which means that $C(\log |\mathcal{X}_2|)$ is no larger
than the GP rate. Thus, we conclude that $C(\log |\mathcal{X}_2|)$
is equal to the GP rate where the state is i.i.d. and uniformly
distributed on set $\mathcal{X}_2$.

\section{Proofs}
\subsection{Proof of Theorem \ref{thmconverse}} \label{proof_thmconverse}
Following from Fano's inequality \cite{Cover:book}, we have
\begin{align}
nR_1 \leq H(Y_1^n)-H(Y_1^n|W_1)+n \epsilon_n \label{rate1}
\end{align}
and
\begin{align}
nR_2
&\leq H(Y_2^n)-H(Y_2^n|W_2)+n \epsilon_n \\
& \leq H(Y_2^n)-H(Y_2^n|W_2,X_2^n)+n \epsilon_n\\
& =H(Y_2^n)-H(Y_2^n|X_2^n)+n \epsilon_n \label{data}\\
& =H(Y_2^n)-\sum_{i=1}^n H(Y_{2i}|X_{2i})+n \epsilon_n
\label{memoryless}
\end{align}
where (\ref{data}) follows from the Markov Chain $ W_2 \rightarrow
X_2^n \rightarrow Y_2^n$, and (\ref{memoryless}) follows from the
memoryless property of the channel $p(y_2|x_2)$.

Applying the technique \cite[page 314, eqn (3.34)]{Csiszar:book}
twice, we obtain
\begin{align}
H(Y_1^n)-H(Y_2^n)&=\sum_{i=1}^n H(Y_{1i}|Y_1^{i-1}, Y_{2(i+1)}^n)-H(Y_{2i}|Y_1^{i-1}, Y_{2(i+1)}^n), \label{one}\\
H(Y_1^n|W_1)-H(Y_2^n|W_1)&=\sum_{i=1}^n H(Y_{1i}|Y_1^{i-1},
Y_{2(i+1)}^n, W_1)-H(Y_{2i}|Y_1^{i-1}, Y_{2(i+1)}^n, W_1).
\label{two}
\end{align}
Define auxiliary random variables as
\begin{align}
V_i=Y_1^{i-1}, Y_{2(i+1)}^n, \qquad i=1,2,\cdots,n.
\end{align}
Further define $Q$ to be an auxiliary random variable that is
independent of everything else and uniform on the set $\{1,
2,\cdots, n\}$, and
\begin{align}
V=(V_Q, Q), \quad U=(V, W_1), \quad X_1=X_{1Q}, \quad X_2=X_{2Q},
\quad Y_1=Y_{1Q}, \quad Y_2=Y_{2Q}. \label{aux}
\end{align}
It is straightforward to check that the random variables thus
defined satisfy (\ref{dis}).

Following from (\ref{one}) and (\ref{two}), we have
\begin{align}
\frac{1}{n} \left( H(Y_1^n)-H(Y_2^n) \right)=H(Y_1|V)-H(Y_2|V) \label{three}\\
\frac{1}{n} \left( H(Y_1^n|W_1)-H(Y_2^n|W_1)
\right)=H(Y_1|U)-H(Y_2|U). \label{four}
\end{align}
Notice that (\ref{three}) implies that there exists a number
$\gamma$ where
\begin{align}
\frac{1}{n} H(Y_1^n)&=H(Y_1|V)+\gamma\\
\frac{1}{n} H(Y_2^n)&=H(Y_2|V)+\gamma \label{SS}\\
0 \leq \gamma & \leq \min \left \{I(V;Y_1), I(V;Y_2) \right\}
\label{range}
\end{align}
where (\ref{range}) follows because $H(Y_1^n) \leq n H(Y_1)$ and
$H(Y_2^n) \leq n H(Y_2)$ and
\begin{align}
H(Y_1^n)=\sum_{i=1}^n H(Y_{1i}|Y_1^{i-1}) \geq \sum_{i=1}^n
H(Y_{1i}|Y_1^{i-1}, Y_{2(i+1)}^n)=n H(Y_1|V)
\end{align}

Following from (\ref{memoryless}), we have
\begin{align}
R_2 & =\frac{1}{n}H(Y_2^n)-\frac{1}{n}\sum_{i=1}^n H(Y_{2i}|X_{2i})+\epsilon_n \nonumber \\
&= H(Y_2|V)+\gamma-H(Y_2|X_2, Q)+\epsilon_n \label{use17} \\
&= H(Y_2|V)+\gamma-H(Y_2|X_2)+\epsilon_n \label{memoryless2}\\
& \leq H(Y_2|V)+\min \left\{ I(V;Y_1), I(V;Y_2) \right \}-H(Y_2|X_2)+\epsilon_n\label{Rs}\\
&=I(X_2;Y_2|V)+ \min \left \{I(V;Y_1), I(V;Y_2) \right \} +
\epsilon_n \label{Markov1}
\end{align}
where (\ref{use17}) follows from (\ref{SS}) and the definition of
the random variables in (\ref{aux}); (\ref{memoryless2}) follows by
the memoryless nature of the channel $p(y_2|x_2)$; (\ref{Rs})
follows from (\ref{range}); and (\ref{Markov1}) follows because the
random variables satisfy (\ref{dis}) which implies that Markov chain
$V \rightarrow X_2 \rightarrow Y_2$ holds.

Following from (\ref{rate1}), we have
\begin{align}
R_1 & \leq \frac{1}{n}H(Y_1^n)-\frac{1}{n}H(Y_1^n|W_1)+\epsilon_n \nonumber\\
& =\frac{1}{n} H(Y_2^n)+H(Y_1|V)-H(Y_2|V)-\frac{1}{n} H(Y_2^n|W_1)-H(Y_1|U)+H(Y_2|U)+\epsilon_n \label{use}\\
& =H(Y_1|V)-H(Y_2|V)-H(Y_1|U)+H(Y_2|U)+\epsilon_n \label{useZ}\\
&=I(U;Y_1|V)-I(U;Y_2|V)+\epsilon_n \label{useMarkov}
\end{align}
where (\ref{use}) follows from (\ref{three}) and (\ref{four});
(\ref{useZ}) follows from the fact that $Y_2^n$ only depends on
$X_2^n$ and the channel noise induced by $p(y_2^n|x_2^n)$, and is
therefore independent of $W_1$; and (\ref{useMarkov})  follows
because the random variables satisfy (\ref{dis}) which implies that
Markov chain $V \rightarrow U \rightarrow (Y_1, Y_2)$ holds.

We obtain the desired upper bound on the capacity region from
(\ref{Markov1}) and (\ref{useMarkov}).

\subsection{Proof of Theorem \ref{thmach}} \label{proof_thmach}
Since the encoding/decoding procedure follows the standard steps,
the detailed calculation of the probability of error is omitted.

\emph{Codebook generation}:

Fix a distribution $p(v, u, x_2)p(x_1|u,x_2)$.

The codebook at Transmitter 2 is generated as follows: generate
$2^{n\gamma}$ sequences $v^n$ in an i.i.d. fashion using $p(v)$.
These $v^n$ sequences constitute the inner codebook. For each $v^n$,
generate $2^{n(R_2-\gamma)}$ sequences $x_2^n$ in an i.i.d. fashion
using $p(x_2|v)$. These $x_2^n$ constitute the outer codebook of
Transmitter 2  associated with $v^n$.

The codebook at Transmitter 1 (the cognitive transmitter) uses the
same inner codebook as Transmitter 2 and the outer codebook of
Transmitter 1 is generated as follows:  for each $v^n$ sequence,
generate $2^{n(R_1+R_0)}$ sequences $u^n$ in an i.i.d. fashion using
$p(u|v)$. These $u^n$ constitute the outer codebook of Transmitter 1
associated with $v^n$ . Randomly distribute them into $2^{nR_1}$
many bins. Each bin will contain approximately $2^{nR_0}$ many $u^n$
sequences.

\emph{Encoding}:

Transmitter 2 splits its message $W_2$ into two independent parts
$W_{2a}$ and $W_{2b}$, with rates $\gamma$ and $R_2-\gamma$,
respectively. For $W_{2a}=w_{2a}$ and $W_{2b}=w_{2b}$, it finds the
$w_{2a}$-th codeword in the inner codebook, denoted as $\bar{v}^n$,
and transmits the $w_{2b}$-th codeword in the outer codebook
(denoted as $\bar{x}_2^n$) of Transmitter 2 associated with
$\bar{v}^n$.

The cognitive encoder knows $W_2$ and therefore knows $\bar{x}_2^n$
and $\bar{v}^n$. For $W_1=w_1$, it look into the $w_1$-th bin in the
outer codebook of Transmitter 1 associated with $\bar{v}^n$, and
find the $u^n$ (denoted as $\bar{u}^n$) that is jointly typical with
$\bar{x}_2^n$ conditioned on $\bar{v}^n$. This can be done almost
always as long as
\begin{align}
R_0 \geq I(U;X_2|V) \label{ach1}
\end{align}
is satisfied. The cognitive encoder then transmit an $x_1^n$
sequence generated i.i.d. conditioned on $\bar{u}^n$ and
$\bar{x}_2^n$ using $p(x_1|u, x_2)$.

\emph{Decoding}: Receiver 2 first finds the unique $v^n$ sequence in
the inner codebook that is jointly typical with received sequence
$y_2^n$ while treating everything else as noise. This can be done if
\begin{align}
\gamma \leq I(V;Y_2) \label{ach2}
\end{align}

Based on the $v^n$ sequence it decoded, Receiver 2 then finds the
unique $x_2^n$ sequence that is jointly typical with $y_2^n$
conditioned on $v^n$ in the outer codebook of Transmitter 2
associated with $v^n$. This can be done if
\begin{align}
R_2-\gamma & \leq I(X_2;Y_2|V) \label{ach3}
\end{align}

Receiver 1 first finds the unique $v^n$ sequence in the inner
codebook that is jointly typical with received sequence $y_1^n$
while treating everything else as noise. This can be done if
\begin{align}
\gamma \leq I(V;Y_1) \label{ach4}
\end{align}

Based on the $v^n$ it decoded, Receiver 1 then finds the unique
$u^n$ sequence that is jointly typical with $y_2^n$  conditioned on
$v^n$ in the outer codebook of Transmitter 1 associated with $v^n$.
This can be done if
\begin{align}
R+R_0 \leq I(U;Y_1|V) \label{ach5}
\end{align}

Based on (\ref{ach1})-(\ref{ach5}), using Fourier-Motzkin
elimination, we obtain the desired result.

\section*{Acknowledgement}
The authors would like to thank Dr. Wei Kang for the helpful
discussions.

\bibliographystyle{unsrt}
\bibliography{ref}

\end{document}